\documentclass[submission,copyright,creativecommons]{eptcs}

\usepackage[british]{babel}
\usepackage[utf8]{inputenc}
\usepackage{amsmath,amssymb,amsfonts,stmaryrd,dsfont,upgreek,mathrsfs}
\usepackage{color}
\usepackage{Common/prooftree}
\usepackage{Common/theorems}

\long\def\ignore#1{\relax}

\newcommand\struto[1][15pt]{{\raise #1 \hbox{\strut}}}%
\newcommand\strutb[1][15pt]{{\raise-#1 \hbox{\strut}}}%

\newcommand\midline[1][5pt]{\\[#1]\hline\struto[5pt]}

\makeatletter
\def\@boxfigurewith[#1]{\figure[#1]\vbox\bgroup\hrule height.1em}
\def\@boxfigurewithout{\figure\vbox\bgroup\hrule height.1em}
\newenvironment{bfigure}{\@ifnextchar[\@boxfigurewith\@boxfigurewithout}{\vskip.2em\hrule height.1em\egroup\endfigure}
\makeatother

 \newcommand\toaux[1]{\immediate\write\@auxout{#1}}

\newcommand\olditem{}
\newcommand\olditemize{}
\newcommand\oldenditemize{}
\newcommand\oldenumerate{}
\newcommand\oldendenumerate{}
\let\olditem\item
\let\olditemize\itemize
\let\oldenditemize\enditemize
\let\oldenumerate\enumerate
\let\oldendenumerate\endenumerate

\newbox\itemlabbox
\newdimen\itemlabwd
\newdimen\itemizespacing
\newcommand\myitem{}
\makeatletter
\def\myitem{%
\@ifnextchar[\@myitemwith\@myitemwithout}
\long\def\@myitemwith[#1]{\hrule height0pt%
  \olditem[]%
  \removelastskip\vskip\itemizespacing%
  \leavevmode%
  \hskip-\leftmargin
  \setbox\itemlabbox=\hbox\bgroup\hfil#1\hfil\egroup%
  \ifdim\wd\itemlabbox>\itemlabwd%
  \box\itemlabbox%
  \else\hbox to\itemlabwd{#1\hfill}\fi%
  \rule[0pt]{0pt}{0pt}}
\long\def\@myitemwithout{\removelastskip\hrule height0pt\olditem\unskip\vskip\itemizespacing\rule[0pt]{0pt}{0pt}}
\makeatother

\renewenvironment{itemize}[1][]{%
  \removelastskip%
  \setbox\itemlabbox=\hbox\bgroup\hfil#1\hfil\egroup%
  \setlength\itemlabwd{\wd\itemlabbox}%
  \def\item{\myitem}%
  \widowpenalty=3000%
  \nobreak%
  \advance\leftmargini-10pt%
  \advance\leftmarginii-5pt%
  \olditemize%
  \unskip%
  \bgroup%
  \setlength\parskip{\itemizespacing}%
  \setlength\topsep{0pt}%
  \setlength\partopsep{0pt}%
  \setlength\parsep{0pt}%
  \setlength\itemsep{0pt}%
  \setlength\itemindent\parindent%
  \setlength\listparindent\parindent%
}
{\removelastskip\hrule height0pt\nobreak\egroup\oldenditemize\unskip\vskip\itemizespacing%
  \def\item{\olditem}%
}

\renewenvironment{enumerate}[1][0]{%
  \removelastskip%
  \setbox\itemlabbox=\hbox\bgroup\hfil#1\hfil\egroup%
  \setlength\itemlabwd{\wd\itemlabbox}%
  \def\item{\myitem}%
  \widowpenalty=3000%
  \nobreak%
  \advance\leftmargini-10pt%
  \advance\leftmarginii-5pt%
  \oldenumerate%
  \unskip%
  \bgroup%
  \setlength\parskip{\itemizespacing}%
  \setlength\topsep{0pt}%
  \setlength\partopsep{0pt}%
  \setlength\parsep{0pt}%
  \setlength\itemsep{0pt}%
  \setlength\itemindent\parindent%
  \setlength\listparindent\parindent%
}
{\removelastskip\hrule height0pt\nobreak\egroup\oldendenumerate\unskip\vskip\itemizespacing%
  \def\item{\olditem}%
}

\newcommand\oldparforMain{}
\let\oldparforMain\par
\newcommand\topequationskip{.2\baselineskip}
\newcommand\botequationskip{.2\baselineskip}

\makeatletter
\renewcommand\[[1][\topequationskip]{\begingroup\let\mysavedpar\par\let\par\oldparforMain\vskip#1\nopagebreak\hbox to\hsize\bgroup\hfil\(}
\renewcommand\][1][\botequationskip]{\)\hfil\egroup\hrule height 0pt\endgroup\@afterindentfalse\vskip#1}
\makeatother

\newbox\columnsbox
\newbox\tmpbox
\newdimen\columnsheight
\newdimen\columnwidth
\newdimen\remainingwidth
\newdimen\textwidthsave
\def\mycolumnsheight{}

\newcommand\columns[1]{%
  \def\mycolumnsheight{}%
  \setlength\remainingwidth\textwidth%
  \setbox\columnsbox=\vbox\bgroup\vskip0pt\vfil\hbox to\textwidth\bgroup#1\egroup\vfil\egroup%
  \columnsheight=\ht\columnsbox%
  \def\mycolumnsheight{to\columnsheight}%
  \hrule height 0pt\vtop{\hbox to\wd\columnsbox\bgroup#1\egroup}%
}

\makeatletter

\def\commonpart{%
  \setlength\columnwidth{\wd\tmpbox}%
  \vtop{\vskip0pt\hbox to\columnwidth{{\box\tmpbox}}}%
  \advance\remainingwidth-\columnwidth%
  \setlength\textwidth\textwidthsave%
  \hsize\textwidthsave%
}
\def\column{\unskip\setlength\textwidthsave\textwidth\@ifnextchar[\@columnwith\@columnwithout}
\long\def\@columnwith[#1]#2{%
  \def\newhsize{#1\dimexpr\textwidth\relax}%
  \hsize\newhsize%
  \ifdim\hsize<0.1pt\hsize\remainingwidth\fi%
  \setlength\textwidth\hsize%
  \setbox\tmpbox=\hbox to\hsize\bgroup\hfil\vtop\mycolumnsheight{\vskip0pt#2\vskip0pt}\hfil\egroup%
  \commonpart%
}
\long\def\@columnwithout#1{%
  \hsize\remainingwidth%
  \setlength\textwidth\hsize%
  \setbox\tmpbox=\hbox\bgroup\vtop\mycolumnsheight{\vskip0pt#1\vskip0pt}\egroup%
  \commonpart%
}
\makeatother

\newcommand\cquad[1][.01]{\column[#1]{}}

\newenvironment{centre}{\begin{center}\unskip}{\end{center}\unskip}

\newcommand\abs[1]{{\left|#1\right|}}        % val.absolue
\renewcommand{\iff}{if and only if}

\newcommand{\ie}{i.e.~}
\newcommand{\eg}{e.g.~}

\newcommand{\resp}{resp.~}

\newcommand\dom[1]{\textsf{Dom}{(#1)}}

\newcommand{\eqdef}{:=\ }
\newcommand{\recdef}{::=\ }
\newcommand{\sqin}{\textsf{\footnotesize{E}}}

\newcommand\powerset[1]{\mathds P(#1)}

\newcommand{\Gam}{\Gamma}

\newcommand{\Del}{\Delta}

\newcommand\LJF{\textsf{LJF}}

\newcommand\LKF{\textsf{LKF}}
\newcommand\LAF[1][]{$\textsf{LAF}_{#1}$}
\newcommand\LAFcf[1][]{$\textsf{LAF}^{\textsf{cf}}_{#1}$}

\newcommand\mathFomega{F_\omega}
\newcommand\Fomega{\ifmmode\mathFomega\else$\mathFomega$\fi}
\newcommand\mathFomegaC{F_\omega^{\mathcal C}}
\newcommand\FomegaC{\ifmmode\mathFomegaC\else$\mathFomegaC$\fi}
\newcommand\mathDNE{\mathrm{DNE}}
\newcommand\DNE{\ifmmode\mathDNE\else$\mathDNE$\fi}

\newcommand\Coq{{\textsf{Coq}}}

\newcommand\pfunspace{\rightharpoonup}

\newcommand\cons[2]{#1\unskip\colon\hskip-.4em\colon\unskip#2}

\newcommand{\sep}{\mbox{\;\rule[-.35\ht\strutbox]{.7pt}{1.3\ht\strutbox}\;}}    % Separators in BNF definitions

\newcommand{\subst}[3]{ \left\{{}^{#3}\hspace{-.2em}\diagup\hspace{-.2em}_{#2} \right\} #1 }

\newcommand{\Rew}[2][]{\stackrel{#1}{\longrightarrow}_{#2}\;}

\newcommand\ssembasis[3]{\left\llbracket{#3}\right\rrbracket_{#1}^{#2}}

\newcommand{\sem}[2][]{{\ssembasis {#1} {} {#2}}}

\newcommand{\SemTy}[2][]{{\ssembasis {#1} {} {#2}}}
\newcommand{\SemTyP}[2][]{{\ssembasis {#1} + {#2}}}
\newcommand{\SemTyN}[2][]{{\ssembasis {#1} - {#2}}}

\newcommand\SemTe[2]{\sem[#2] {#1}}

\newcommand\orth[2]{{#1\mathrel{{\perp}}#2}}

\newcommand\col{\unskip{\colon}\hskip-.2em}
\newcommand\XcolY[2]{#1\col#2}
\newcommand\varRead[2][]{#2\left[#1\right]}

\newcommand\unitt{\texttt{unit}}
\newcommand\uniti{()}

\newcommand{\seqg}[3]{\mbox{$\ {#1}_{#2}^{#3}\ $}}

\newcommand{\seqf} [2][]{\seqg{\vdash}     {#1}{#2}}

\newcommand{\decf} [2][]{\seqg{\Vdash}    {#1}{#2}}

\newcommand\DerF[4][]{{#2}\seqf[{#1}]{}{[#3]}{#4}}
\newcommand\Der[3][]{{#2}\seqf[{#1}]{}{#3}}

\newcommand\DerDec[4][]{{#3}\decf[{#1}]{#2}{#4}}

\newcommand\DerFLKF[3][]{\seqf[{#1}]{} #2\Downarrow #3}
\newcommand\DerLKF[3][]{\seqf[{#1}]{}#2\Uparrow #3}
\newcommand\DerFlLJF[4][]{[#2]\stackrel{#3}{\longrightarrow}[#4]}
\newcommand\DerFrLJF[3][]{[#2]{-}_{#3}{\rightarrow}}

\newcommand\DerLJF[4][]{[#2]#3\longrightarrow [#4]}

\newcommand\cut{\textsf{cut}}

\newcommand\daggerL{\raise3pt\hbox{\rotatebox{-40}{$\dagger$}}}
\newcommand\daggerR{\raise0pt\hbox{\rotatebox{40}{$\dagger$}}}

\newcommand\andP{{\wedge^+}}
\newcommand\andN{{\wedge^-}}
\newcommand\ou{{\vee}}
\newcommand\orP{{\vee^+}}
\newcommand\orN{{\vee^-}}
\newcommand\trueP{{\top^+}}
\newcommand\trueN{{\top^-}}
\newcommand\falseP{{\bot^+}}
\newcommand\falseN{{\bot^-}}

\newcommand\imp{{\Rightarrow}}

\newcommand{\non}[1]{{#1}^{\perp}}

\newcommand\RCextend[1][]{;}

\newcommand\DecompType[1][]{\mathbb D_{#1}}
\newcommand\Dstruct{\DecompType[\textsf{st}]}

\newcommand\SubsType[1][]{\bullet\bullet\bullet}

\newcommand\var[1][]{\textsf{Lab}_{#1}}
\newcommand\vare[1][]{\var[e]}

\newcommand\domP[1]{\textsf{dom}^+(#1)}
\newcommand\domN[1]{\textsf{dom}^-(#1)}

\newcommand\Contexts[1][]{\mathcal G_{#1}}
\newcommand\Cextend[1][]{;}

\newcommand\TContexts[1][]{\mathsf{Co}_{#1}}
\newcommand\Textend[2][]{;#2}

\newcommand\Atms[1][]{\mathbb A_{#1}}

\newcommand\Moles[1][]{\mathbb M_{#1}}
\newcommand\TDecs[1][]{\mathbb D_{#1}}

\newcommand{\Lit}{\mathbb L}

\newcommand{\Data}{\mathsf{Pat}}
\newcommand{\Datast}[1]{\abs {#1}}
\newcommand{\PTerms}{\mathsf{Terms}}
\newcommand{\Decomp}{\mathsf{Terms}^{\mathsf d}}

\newcommand\Drefute[1]{{\sim}#1}
\newcommand\Dand{,}
\newcommand\Dunit{\bullet}

\newcommand\Tunit{\bullet}
\newcommand\Tand{,}

\newcommand\THO[2][]{{#2}}

\newcommand\cutc[2]{\left\langle#1\mid#2\right\rangle}
\newcommand{\Acc}[2][]{\textsf{Im}^{#1}(#2)}

\newcommand\Ppos{\_^+}
\newcommand\Pneg{\_^-}
\newcommand\Ptrue{\bullet}
\newcommand\inj[2]{\textsf{inj}_{#1}{(#2)}}
\newcommand\paire[2]{({#1},{#2})}

\newcommand\project[2]{\pi_{#1}{(#2)}}

\newcommand\switchr[1]{{\curvearrowright}{(#1)}}

\newcommand\lefths{\textsf{l}}
\newcommand\righths{\textsf{r}}

\newcommand\SPrim{\mathscr L}
\newcommand\SPos{\mathscr P}
\newcommand\SNeg{\mathscr N}

\newcommand\spat[1]{\tilde{#1}}
\newcommand\SContexts{\tilde{\TContexts}}
\newcommand\SCextend[1][]{;}

\newcommand\SDecs{\tilde{\TDecs}}

\newcommand\MolesIneq{\triangleleft}

\newcommand\machine[2][]{\langle\hskip-.5em\langle\ #2\ \mid\ #1\ \rangle\hskip-.5em\rangle}

\renewcommand\sqin{\raise1pt\hbox{\large $\,\epsilon\,$}}

\newcommand\Index[2][]{{\em #2}}

\begin{document}

\title{Realisability semantics of abstract focussing, formalised}

\author{Stéphane Graham-Lengrand
  \institute{CNRS, École Polytechnique, INRIA, SRI International}
}

\newcommand\titlerunning{Realisability semantics of abstract focussing, formalised}
\newcommand\authorrunning{Stéphane Graham-Lengrand}

\maketitle

\begin{abstract}
  We present a sequent calculus for abstract focussing, equipped with
  proof-terms: in the tradition of Zeilberger's work, logical
  connectives and their introduction rules are left as a parameter of
  the system, which collapses the synchronous and asynchronous phases
  of focussing as macro rules. We go further by leaving as a parameter
  the operation that extends a context of hypotheses with new ones,
  which allows us to capture both classical and intuitionistic
  focussed sequent calculi.

  We then define the realisability semantics of (the proofs of) the
  system, on the basis of Munch-Maccagnoni's orthogonality models for
  the classical focussed sequent calculus, but now operating at the
  higher level of abstraction mentioned above. We prove, at that
  level, the Adequacy Lemma, namely that if a term is of type A, then
  in the model its denotation is in the (set-theoretic) interpretation
  of A.  This exhibits the fact that the universal quantification
  involved when taking the orthogonal of a set, reflects in the
  semantics Zeilberger's universal quantification in the macro rule
  for the asynchronous phase.

  The system and its semantics are all formalised in Coq.
\end{abstract}

\section{Introduction}\label{sec:intro}

The objective of this paper is to formalise a strong connection
between \emph{focussing} and \emph{realisability}. 

Focussing is a concept from proof theory that arose from the study of
Linear Logic~\cite{girard-ll,andreoli92focusing} with motivations in
proof-search and logic programming, and was then used for studying the
proof theory of classical and intuitionistic
logics~\cite{Gir:newclc,DJS95,DL:JLC07,liang09tcs}.

Realisability is a concept used since Kleene~\cite{KleeneSC:intint} to
witness provability of formulae and build models of their
proofs. While originally introduced in the context of constructive
logics, the methodology received a renewed attention with the concept
of \emph{orthogonality}, used by Girard to build models of Linear
Logic proofs, and then used to define the realisability semantics of
classical proofs~\cite{DanosKrivine00}.

Both focussing and realisability exploit the notion of \emph{polarity}
for formulae, with an asymmetric treatment of positive and negative
formulae.

In realisability, a primitive interpretation of a positive formula
such as $\exists x A$ (\resp $A_1\vee A_2$) is given as a set of pairs
$(t,\pi)$, where $t$ is a witness of existence and $\pi$ is in the
interpretation of $\subst A x t$ (\resp a set of injections $\inj
i\pi$, where $\pi$ is in the interpretation of $A_i$). In other words,
the primitive interpretation of positive formulae expresses a very
contructive approach to provability. In contrast, a negative formula
(such as $\forall x A$ or $A_1\Rightarrow A_2$) is interpreted ``by
duality'', as the set that is \emph{orthogonal} to the interpretation
of its (positive) negation.  In classical realisability, a
bi-orthogonal set completion provides denotations for all classical
proofs of positive formulae.

In focussing, introduction rules for connectives of the same polarity
can be chained without loss of generality: for instance if we decide
to prove $(A_1\vee A_2)\vee A_3$ by proving $A_1\vee A_2$, then we can
assume (without losing completeness) that a proof of this in turn
consists in a proof of $A_1$ or a proof of $A_2$ (rather than uses a
hypothesis, uses a lemma or reasons by contradiction).

Such a grouping of introduction steps for positives (\resp negatives)
is called a \emph{synchronous} phase (\resp \emph{asynchronous}
phase).
In order to express those phases as (or collapse them into) single
\emph{macro steps}, some formalisms for \emph{big-step focussing} (as
in Zeilberger's work~\cite{ZeilbergerPOPL08,Zeilberger08}) abstract
away from logical connectives and simply take as a parameter the
mechanism by which a positive formula can be decomposed into a
collection of positive atoms and negative formulae. While the proof of
a positive needs to exhibit such a decomposition, the proof of a
negative needs to range over such decompositions and perform a case
analysis.

This asymmetry, and this universal quantification over decompositions,
are reminiscent of the orthogonal construction of realisability
models.  To make the connection precise, we formalise the construction
of realisability models for (big-step) focussed systems in
Zeilberger's style.

In~\cite{MunchCSL09}, the classical realisability semantics
was studied for the classical sequent calculus, with an extended
notion of cut-elimination that produced focussed proofs.  Here we want
to avoid relying on the specificities of particular logical
connectives, and therefore lift this realisability semantics to the
more abstract level of Zeilberger's systems, whose big-step approach
also reveals the universal quantification that we want to connect to
the orthogonality construction. We even avoid relying on the
specificities of a particular logic as follows:
\begin{itemize}
\item We do not assume that formulae are syntax, \ie have an inductive
  structure, nor do we assume that ``positive atoms'' are
  particular kinds of formulae; positive atoms and formulae can now
  literally be two arbitrary sets, of elements called
  \Index[atom]{atoms} and \Index[molecule]{molecules}, respectively.
\item The operation that extends a context of hypotheses with new
  ones, is usually taken to be set or multiset union. We leave this
  extension operation as another parameter of the system, since
  tweaking it will allow the capture of different logics.
\end{itemize}

In Section~\ref{sec:LAFDefQF} we present our abstract system called
\LAF, seen through the Curry-Howard correspondence as a typing system
for (proof-)terms. Section~\ref{sec:LAFexprop} describes how to tune
(\ie instantiate) the notions of atoms, molecules, and context
extensions, so as to capture the big-step versions of standard
focussed sequent calculi, both classical and
intuitionistic. Section~\ref{sec:RealTerms} gives the definition of
realisability models, where terms and types are interpreted, and
proves the Adequacy Lemma. In very generic terms, if $t$ is a
proof(-term) for $A$ then in the model the interpretation of $t$ is in
the interpretation of $A$.  Finally, Section~\ref{sec:RealConsistency}
exhibits a simple model from which we derive the consistency of \LAF.

\section{The abstract focussed sequent calculus \LAF}
\label{sec:LAFDefQF}

An instance of \LAF\ is given by a tuple of parameters
\((\var[+],\var[-],\Atms,\Moles,\TContexts,\Data,\decf{})\) where each
parameter is described below.  We use $\rightarrow$ (\resp
$\pfunspace$) for the total (\resp partial) function space
constructor.

Since our abstract focussing system is a typing system, we use a
notion of typing contexts, \ie those structures denoted $\Gamma$ in a
typing judgement of the form $\Der\Gamma\ldots{}$. Two kinds of
``types'' are used (namely atoms and molecules), and what is declared
as having such types in a typing context, is two corresponding kinds
of labels:\footnote{We choose to call them ``labels'', rather than
  ``variables'', because ``variable'' suggests an object identified by
  a name that ``does not matter'' and somewhere subject to
  $\alpha$-conversion. Our labels form a deterministic way of indexing
  atoms and molecules in a context, and could accommodate De Bruijn's
  indices or De Bruijn's levels.}  \Index[positive label]{positive
  labels} and \Index[negative label]{negative labels}, respectively
ranging over $\var[+]$ and $\var[-]$. These two sets are the first two
parameters of \LAF.

% Contexts will be extendable (\eg we may want to extend $\Gamma$ with a
% new type declaration), and the following data-structure formalises
% what contexts will be extended with.

\begin{definition}[Generic contexts, generic decomposition algebras]

  Given two sets $\mathcal A$ and $\mathcal B$, an \Index[context
    algebra]{$(\mathcal A,\mathcal B)$-context algebra} is an algebra
  of the form
  \[\left(
  \Contexts,
  {\footnotesize
    \left(\begin{array}{c@{}l@{}l}
      \Contexts\times\var[+]&\pfunspace&\mathcal A\\
      (\Gam,x^+)&\mapsto&\varRead[x^+]\Gam
    \end{array}\right),
    \left(\begin{array}{c@{}l@{}l}
      \Contexts\times\var[-]&\pfunspace&\mathcal B\\
      (\Gam,x^-)&\mapsto&\varRead[x^-]\Gam
    \end{array}\right),
    \left(\begin{array}{c@{}l@{}l}
      \Contexts\times\DecompType[\mathcal A,\mathcal B]&\rightarrow&\Contexts\\
      (\Gam,\Delta)&\mapsto&\Gam\Cextend\Delta
    \end{array}\right)
  }
  \right)
  \]
  whose elements are called $(\mathcal A,\mathcal
  B)$-\Index[context]{contexts}, and where $\DecompType[\mathcal
  A,\mathcal B]$, which we call the \Index{$(\mathcal A,\mathcal
    B)$-decomposition algebra} and whose elements are called
  $(\mathcal A,\mathcal B)$-\Index[decomposition]{decompositions}, is
  the free algebra defined by the following grammar:
  \[\Delta,\Delta_1,\ldots\recdef a\sep \Drefute b\sep\Dunit\sep\Delta_1\Dand\Delta_2\]
  where $a$ (\resp $b$) ranges over $\mathcal A$ (\resp $\mathcal B$).

  Let $\Dstruct$ abbreviate $\DecompType[\unitt,\unitt]$, whose
  elements we call \Index[decomposition structure]{decomposition
    structures}.

  The \Index[structure]{(decomposition) structure} of an $(\mathcal
  A,\mathcal B)$-decomposition $\Delta$, denoted $\abs\Delta$, is its
  obvious homomorphic projection in $\Dstruct$.

  We denote by $\domP\Gamma$ (\resp $\domN\Gamma$) the subset of
  $\var[+]$ (\resp $\var[-]$) where $\varRead[x^+]\Gam$ (\resp
  $\varRead[x^-]\Gam$) is defined, and say that $\Gamma$ is
  \Index{empty} if both $\domP\Gamma$ and $\domN\Gamma$ are.
\end{definition}

Intuitively, an $(\mathcal A,\mathcal B)$-decomposition $\Delta$ is
simply the packaging of elements of $\mathcal A$ and elements of
$\mathcal B$; we could flatten this packaging by seeing $\Dunit$ as
the empty set (or multiset), and $\Delta_1\Dand\Delta_2$ as the union
of the two sets (or multisets) $\Delta_1$ and $\Delta_2$.
Note that the coercion from $\mathcal B$ into $\DecompType[\mathcal
A,\mathcal B]$ is denoted with $\Drefute{}$. It helps
distinguishing it from the coercion from $\mathcal A$ (\eg when
$\mathcal A$ and $\mathcal B$ intersect each other), and in many
instances of \LAF\ it will remind us of the presence of an otherwise
implicit negation. But so far it has no logical meaning, and in
particular $\mathcal B$ is not equipped with an operator $\Drefute{}$
of syntactical or semantical nature.

Now we can specify the nature of the \LAF\ parameters:
\begin{definition}[\LAF\ parameters]
  Besides $\var[+]$ and $\var[-]$, \LAF\ is parameterised by
  \begin{itemize}
  \item two sets $\Atms$ and $\Moles$, whose elements are respectively
    called \Index[atom]{atoms} (denoted $a$, $a'$,\ldots), and
    \Index[molecule]{molecules} (denoted $M$, $M'$,\ldots);
  \item an $(\Atms,\Moles)$-context algebra $\TContexts$,
    whose elements are called \Index[typing context]{typing contexts};
  \item a \Index{pattern algebra}, an algebra of the form
    \[\left(\Data,
      {\footnotesize
        \left(\begin{array}{c@{}l@{}l}
            \Data&\rightarrow&\Dstruct\\
            p&\mapsto&\Datast p
          \end{array}\right)
      }
    \right)
    \] 
    whose elements are called \Index[pattern]{patterns},
  \item a \Index{decomposition relation}, \ie a set of elements
    \[(\DerDec{}{\_}{\XcolY \_{\_}}):(\DecompType\times\Data\times\Moles)\]
    such that if $\DerDec{}{\Delta}{\XcolY p{M}}$ then the structure of
    $\Delta$ is $\Datast p$.
  \end{itemize}

  The $(\Atms,\Moles)$-decomposition algebra, whose elements are
  called \Index{typing decompositions}, is called the \Index[typing
  decomposition algebra]{typing decomposition algebra} and is denoted
  $\DecompType$.
\end{definition}

The group of parameters \((\Atms,\Moles)\) specifies what the instance
of \LAF, as a logical system, talks about. A typical example is when
\(\Atms\) and \(\Moles\) are respectively the sets of (positive) atoms
and the set of (positive) formulae from a polarised
logic. Section~\ref{sec:LAFexprop} shows how our level of abstraction allows for
some interesting variants. In the Curry-Howard view, \(\Atms\) and
\(\Moles\) are our sets of types.

The last group of parameters \((\Data,\decf{})\) specifies the
structure of molecules. If \(\Moles\) is a set of formulae featuring
logical connectives, those parameters specify the introduction rules
for the connectives.
The intuition behind the terminology is that the decomposition
relation $\decf{}$ decomposes a molecule, according to a pattern, into
a typing decomposition which, as a first approximation, can be seen as
a ``collection of atoms and (hopefully smaller) molecules''.

\begin{definition}[\LAF\ system]
  Proof-terms are defined by the following syntax:
  \[
  \begin{array}{lll@{\recdef}l}
    \mbox{Positive terms }&\PTerms^+&t^+&pd\\
%    \mbox{Negative terms }&\PTerms^-&t^-&x^- \sep f\\
    \mbox{Decomposition terms }&\Decomp&d&x^+ \sep \THO f\sep\Tunit\sep d_1\Tand d_2 \\
    \mbox{Commands}&\PTerms&c& \cutc{x^-}{t^+} \sep \cutc{f}{t^+}
  \end{array}
  \]
  where $p$ ranges over $\Data$, $x^+$ ranges over $\var[+]$, $x^-$ ranges over $\var[-]$, and $f$ ranges over the partial function space $\Data\pfunspace\PTerms$.

  \LAF\ is the inference system of Fig.~\ref{def:LAFqf} defining the derivability of three kinds of sequents
  \[
  \begin{array}{l@{\quad:\quad}l}
    (\DerF\_ {\XcolY {\_} \_}{}) & (\TContexts\times \PTerms^+\times\Moles)\\
    (\Der\_ {\XcolY \_ \_}) & (\TContexts\times \Decomp\times\DecompType)\\
    (\Der\_ \_) & (\TContexts\times \PTerms)
  \end{array}
  \]
  We further impose in rule \textsf{async} that the domain of function $f$ be exactly those patterns that can decompose $M$ ($p\in\dom f$ \iff\ there exists $\Delta$ such that $\DerDec{}{\Delta}{\XcolY p{M}}$).

  \LAFcf\ is the inference system \LAF\ without the \cut-rule.
\end{definition}
  \begin{bfigure}[!h]
    \[
    \begin{array}{c}
      \infer[\textsf{sync}]{
        \DerF\Gam {\XcolY {p d} M}{}
      }{
        \DerDec{}{\Delta}{\XcolY p{M}}
        \quad
        \Der{\Gamma}{\XcolY  d{\Delta}}
      }
      \\[-10pt]
      \midline
      \\[-10pt]
      \infer{\Der\Gam{\XcolY{\Tunit}{\Dunit}}}{\strut}
      \qquad
      \infer{
        \Der\Gam{\XcolY{ d_1\Tand d_2}{\Del_1\Dand\Del_2}}
      }{
        \Der\Gam{\XcolY{ d_1}{\Del_1}}
        \quad
        \Der\Gam{\XcolY{ d_2}{\Del_2}}
      }
      \\\\
      \infer[\textsf{init}]{
        \Der{\Gam}{\XcolY{x^+}{a}}}{\varRead[x^+]\Gam = a
      }
      \qquad
      \infer[\textsf{async}]{
        \Der\Gam {\XcolY {\THO f} {\Drefute M}}
      }{
        \forall p,\forall\Delta,\quad\DerDec{}{\Delta}{\XcolY p{M}}\quad\imp\quad\Der{\Gamma\Cextend[x]\Delta}{f(p)}
      }
      \\[-10pt]
      \midline
      \\[-10pt]
      \infer[\textsf{select}]{\Der\Gam{\cutc{x^-}{t^+}}}{
        \DerF\Gam {\XcolY{t^+}{\varRead[x^-]\Gam}} {}
      }
      \qquad
      \infer[\cut]{\Der\Gam{\cutc{f}{t^+}}}{
        \Der\Gam {\XcolY {f}{\Drefute M}}\qquad
        \DerF\Gam {\XcolY{t^+}M} {}
      }
    \end{array}
    \]
    \caption{\LAF}
    \label{def:LAFqf}
  \end{bfigure}

An intuition of \LAF\ can be given in terms of proof-search:

When we want to ``prove'' a molecule, we first need to decompose it
into a collection of atoms and (refutations of) molecules (rule
$\textsf{sync}$). Each of those atoms must be found in the current
typing context (rule $\textsf{init}$). Each of those molecules must be
refuted, and the way to do this is to consider all the possible ways
that this molecule could be decomposed, and for each of those
decompositions, prove the inconsistency of the current typing context
extended with the decomposition (rule $\textsf{async}$). This can be
done by proving one of the molecules refuted in the typing context
(rule $\textsf{select}$) or refuted by a complex proof (rule
$\textsf{cut}$). Then a new cycle starts.

Now it will be useful to formalise the idea that, when a molecule $M$ is
decomposed into a collection of atoms and (refutations of) molecules,
the latter are ``smaller'' than $M$:
\begin{definition}[Well-founded \LAF\ instance]

  We write $M'\MolesIneq M$ if there are $\Delta$ and $p$
  such that $\DerDec{}{\Delta} {\XcolY p {M}}$ and $\Drefute{M'}$ is a
  leaf of $\Delta$.

  The \LAF\ instance is \Index[well-founded
    \LAF\ instance]{well-founded} if (the smallest order containing)
  $\MolesIneq$ is well-founded.
\end{definition}
Well-foundedness is a property that a \LAF\ instance may or may not
have, and which we will require to construct its realisability
semantics. A typical situation where it holds is when $M'\MolesIneq M$
implies that $M'$ is a sub-formula of $M$.

% Typing decompositions and decomposition terms organise the packaging
% of the proofs of atoms and (refuted) molecules decomposed by rule
% $\textsf{sync}$. Typing decompositions could here be taken to be a
% multiset of atoms and (refuted) molecules, but keeping a dedicated
% structure for the packaging will be more convenient when we add
% quantifiers: giving decompositions an inductive structure allows a
% lossless modelling of quantifiers' scopes.

The above intuitions may become clearer when we instantiate the
parameters of \LAF\ with actual literals, formulae, etc in order to
capture existing systems: we shall therefore we illustrate system
\LAF\ by specifying different instances, providing each time the long
list of parameters, that capture different focussed sequent calculus
systems.

While \LAF\ is defined as a typing system (in other words with
proof-terms decorating proofs in the view of the Curry-Howard
correspondence), the traditional systems that we capture below are
purely logical, with no proof-term decorations. When encoding the
former into the latter, we therefore need to erase proof-term
annotation, and for this it is useful to project the notion of typing
context as follows:

\begin{definition}[Referable atoms and molecules]
  Given a typing context $\Gamma$, let $\Acc[+]\Gamma$ (\resp
  $\Acc[-]\Gamma$) be the image of function
  $x^+\mapsto\varRead[x^+]\Gamma$ (\resp
  $x^-\mapsto\varRead[x^-]\Gamma$), \ie the set of atoms (\resp
  molecules) that can be refered to, in $\Gamma$, by the use of a
  positive (\resp negative) label.
\end{definition}

\section{Examples in propositional logic}\label{sec:LAFexprop}

The parameters of \LAF\ will be specified so as to capture: the
classical focussed sequent calculus \LKF\ and the
intuitionistic one \LJF~\cite{liang09tcs}.

\subsection{Polarised classical logic - one-sided}
\label{sec:LAFK1qf}

In this sub-section we define the instance \LAF[K1] corresponding to
the (one-sided) calculus \LKF:

\begin{definition}[Literals, formulae, patterns, decomposition]
  \label{def:classpolsyntax}

  Let $\Atms$ be a set of elements called \emph{atoms} and ranged over by $a, a',\ldots$.

  \emph{Negations} of atoms $\non a, \non{a'},\ldots$ range over a set isomorphic to, but disjoint from, $\Atms$.

  Let $\Moles$ be the set defined by the first line of the following grammar for (polarised) formulae of classical logic:
  \[ 
  \begin{array}{llll}
    \mbox{Positive formulae}&  P,\ldots&\recdef a\sep\trueP \sep \falseP \sep A\andP B\sep A \orP B\\
    \mbox{Negative formulae}& N,\ldots&\recdef \non a\sep\trueN \sep \falseN \sep A\andN B\sep A \orN B\\
    \mbox{Unspecified formulae}&A&\recdef P \sep N
  \end{array}
  \]

  Negation is extended to formulae as usual by De Morgan's laws (see \eg\cite{liang09tcs}).
  % \[ \begin{array}{llll}
  %   \non{\trueP}&\eqdef\falseN&\non{\trueN}&\eqdef\falseP\\
  %   \non{\falseP}&\eqdef\trueN&\non{\falseN}&\eqdef\trueP\\
  %   \non{(A\andP B)}&\eqdef\non A\orN \non B&\non{(A\andN B)}&\eqdef\non A\orP \non B\\
  %   \non{(A\orP B)}&\eqdef\non A\andN \non B&\non{(A\orN B)}&\eqdef\non A\andP \non B
  % \end{array}
  % \]
  %% and we extend it to sets or multisets of formulae pointwise.

  The set $\Data$ of \Index[pattern]{pattern} is defined by the following grammar:
  \[p,p_1,p_2,\ldots\recdef \Ppos \sep \Pneg \sep \Ptrue\sep\paire {p_1}{p_2}\sep\inj i p\]

  The decomposition relation $(\DerDec{}{\_}{\XcolY
    \_{\_}}):(\DecompType\times\Data\times\Moles)$ is the restriction
  to molecules of the relation defined inductively for all formulae
  by the inference system of Fig.~\ref{fig:decompK1}.

  The map $p\mapsto \Datast p$ can be inferred from the decomposition relation.
\end{definition}
\begin{bfigure}
    \[
    \begin{array}c
      \infer{\DerDec{}{\Dunit}{\XcolY{\Ptrue}\trueP}}{}
      \qquad
      \infer{\DerDec{}{\Drefute {\non N}}{\XcolY{\Pneg}{N}}}{}
      \qquad
      \infer{\DerDec{}{a}{\XcolY{\Ppos}{a}}}{}\\\\
      \infer{\DerDec{}{\Delta_1,\Delta_2}{\XcolY{\paire{p_1}{p_2}}A_1\andP A_2}}{
        \DerDec{}{\Delta_1}{\XcolY{p_1}A_1}
        \quad
        \DerDec{}{\Delta_2}{\XcolY{p_2}A_2}
      }
      \qquad
      \infer{\DerDec{}{\Delta}{\XcolY{\inj i p}A_1\orP A_2}}{\DerDec{}{\Delta}{\XcolY{p}A_i}}  
    \end{array}
    \]
  \caption{Decomposition relation for \LAF[K1]}
  \label{fig:decompK1}
\end{bfigure}

Keeping the \textsf{sync} rule of \LAF[K1] in mind, we can already see
in Fig.~\ref{fig:decompK1} the traditional introduction rules of
positive connectives in polarised classical logic.  Note that these
rules make \LAF[K1] a well-founded \LAF\ instance, since $M'\MolesIneq
M$ implies that $M'$ is a sub-formula of $M$.  The rest of this
sub-section formalises that intuition and explains how \LAF[K1]
manages the introduction of negative connectives, etc.  But in order
to finish the instantiation of \LAF\ capturing \LKF, we need to define
typing contexts, \ie give $\var[+]$, $\var[-]$, and $\TContexts$. In
particular, we have to decide how to refer to elements of the typing
context. To avoid getting into aspects that may be considered as
implementation details (in~\cite{LengrandHDR} we present two
implementations based on De Bruijn's indices and De Bruijn's levels),
we will stay rather generic and only assume the following property:

\begin{definition}[Typing contexts]\label{def:classicalcontext} We assume that context extensions satisfy:
  \[
  \begin{array}{ll@{\qquad}ll}
    \Acc[+]{\Gamma\Cextend a}&=\Acc[+]{\Gamma}\cup \{a\}&  \Acc[-]{\Gamma\Cextend a}&=\Acc[-]{\Gamma}\\
    \Acc[+]{\Gamma\Cextend \Drefute M}&=\Acc[+]{\Gamma}&  \Acc[-]{\Gamma\Cextend \Drefute M}&=\Acc[-]{\Gamma}\cup \{M\}\\
    \Acc[\pm]{\Gamma\Cextend\Dunit}&=\Acc[\pm]{\Gamma}&
    \Acc[\pm]{\Gamma\Cextend(\Delta_1\Dand\Delta_2)}&=\Acc[\pm]{\Gamma\Cextend\Delta_1\Cextend\Delta_2}
  \end{array}
  \]
  where $\pm$ stands for either $+$ or $-$.
\end{definition}

% Notice that we have, for all $\Gamma$ and $\Delta$,
% \[
% \Acc[+]{\Gamma\Cextend\Delta}\cup\non{\Acc[-]{\Gamma\Cextend\Delta}} = \Acc[+]{\Gamma}\cup\non{\Acc[-]{\Gamma}}\cup\flatten\Delta
% \]
% \noindent where $\flatten\Delta$ is the obvious flattening of $\Delta$ as a multiset of atoms and molecules.

We now relate (cut-free) \LAFcf[K1] and the \LKF\ system
of~\cite{liang09tcs} by mapping sequents:

% \begin{definition}[Flattening typing decompositions]
%   Let $\flatten\Delta$ be the flattening of a typing decomposition as a multiset of positive literals and negative formulae, \ie
%   \[\begin{array}{ll@{\qquad}ll}
%   \flatten a&\eqdef \multiset a&\flatten {\Drefute P}&\eqdef \multiset{\non P}\\
%   \flatten \Dunit&\eqdef \emptyset&\flatten {\Delta_1\Dand \Delta_2}&\eqdef \flatten {\Delta_1}\cup\flatten {\Delta_2}
%   \end{array}
%   \]
% \end{definition}

\begin{definition}[Mapping sequents]

  We encode the sequents of \LAF[K1] (regardless of derivability) to those of \LKF\ as follows:
  \[\begin{array}{l@{\qquad\eqdef\qquad}l}
  \phi(\Der\Gamma {c})&\DerLKF {\non{\Acc[+]\Gamma},{\Acc[-]\Gamma}}{}\\
  \phi(\Der\Gamma {\XcolY {x^+} {a}})&\DerFLKF{\non{\Acc[+]\Gamma},{\Acc[-]\Gamma}} {a}\\
  \phi(\Der\Gamma {\XcolY {f} {\Drefute P}})&\DerFLKF {\non{\Acc[+]\Gamma},{\Acc[-]\Gamma}}{\non P}\\
  \phi(\DerF\Gamma {\XcolY {t^+} P}{})&\DerFLKF{\non{\Acc[+]\Gamma},{\Acc[-]\Gamma}} {P}
  \end{array}
  \]
\end{definition}

\begin{theorem}\label{th:adequacyK1} $\phi$ satisfies structural adequacy between \LAFcf[K1]\ and \LKF.
\end{theorem}
The precise notion of adequacy used here is formalised
in~\cite{LengrandHDR}; let us just say here that it preserves the
derivability of sequents in a compositional way (a derivation $\pi$ in one
system is mapped to a derivation $\pi'$ in the other system, and its
subderivations are mapped to subderivations of $\pi'$).

\subsection{Polarised intuitionistic logic}

In this sub-section we define the instance \LAF[J] corresponding to
the (two-sided) calculus \LJF. Because it it two-sided, and the \LAF\
framework itself does not formalise the notion of side (it is not
incorrect to see \LAF\ as being intrinsically one-sided), we shall
embed a \emph{side information} in the notions of atoms and molecules:

\begin{definition}[Literals, formulae, patterns, decomposition]\nopagesplit

  Let $\Lit^+$ (\resp $\Lit^-$) be a set of elements called positive (\resp negative) literals, and ranged over by $l^+,l_1^+,l_2^+,\ldots$ (\resp $l^-,l^-_1,l^-_2,\ldots$).
  Formulae are defined by the following grammar:
  \[ 
  \begin{array}{llll}
    \mbox{Positive formulae}&  P,\ldots&\recdef l^+\sep\trueP \sep \falseP \sep A\andP B\sep A \ou B\\
    \mbox{Negative formulae}& N,\ldots&\recdef l^-\sep\trueN \sep \falseN \sep A\andN B\sep A \imp B\sep \neg A\\
    \mbox{Unspecified formulae}&A&\recdef P \sep N
  \end{array}
  \]

  We \emph{position} a literal or a formula on the left-hand side or the right-hand side of a sequent by combining it with an element, called \Index{side information}, of the set $\{\lefths,\righths\}$: we define
  \[\begin{array}{ll}
  \Atms&\eqdef\{(l^+,\righths)\sep l^+\mbox{ positive literal}\}\cup\{(l^-,\lefths)\sep l^-\mbox{ negative literal}\}\cup\{(\falseN,\lefths)\}\\
  \Moles&\eqdef\{(P,\righths)\sep P\mbox{ positive formula}\}\cup\{(N,\lefths)\sep N\mbox{ negative formula}\}
  \end{array}
  \]
  In the rest of this sub-section $v$ stands for either a negative literal $l^-$ or $\falseN$.

  The set $\Data$ of \Index[pattern]{pattern} is defined by the following grammar:
  \[\begin{array}{lrl}
  p,p_1,p_2,\ldots&\recdef& \Ppos_\righths \sep \Pneg_\righths \sep \Ptrue_\righths\sep\paire {p_1}{p_2}\sep\inj i p\\
  &\sep&\Ppos_\lefths \sep \Pneg_\lefths \sep \Ptrue_\lefths\sep\cons {p_1}{p_2}\sep\project i p\sep\switchr p
  \end{array}
  \]

  The decomposition relation $(\DerDec{}{\_}{\XcolY
    \_{\_}}):(\DecompType\times\Data\times\Moles)$ is the restriction
  to molecules of the relation defined inductively for all
  positioned formulae by the inference system of Fig.~\ref{fig:decompJ}.
\end{definition}
\begin{bfigure}[!h]
    \[
    \begin{array}c
      \infer{\DerDec{}{\Drefute {(N,\lefths)}}{\XcolY{\Pneg_\righths}{(N,\righths)}}}{\strut}
      \qquad
      \infer{\DerDec{}{(l^+,\righths)}{\XcolY{\Ppos_\righths}{(l^+,\righths)}}}{\strut}
      \\\\
      \infer{\DerDec{}{\Dunit}{\XcolY{\Ptrue_\righths}{(\trueP,\righths)}}}{\strut}
            \\\\
            \infer{\DerDec{}{\Delta_1,\Delta_2}{\XcolY{\paire{p_1}{p_2}}{(A_1\andP A_2,\righths)}}}{
              \DerDec{}{\Delta_1}{\XcolY{p_1}{(A_1,\righths)}}
              \quad
              \DerDec{}{\Delta_2}{\XcolY{p_2}{(A_2,\righths)}}
            }
            \qquad
            \infer{\DerDec{}{\Delta}{\XcolY{\inj i p}{(A_1\ou A_2,\righths)}}}{\DerDec{}{\Delta}{\XcolY{p}{(A_i,\righths)}}}\\\\
            \infer{\DerDec{}{\Drefute {(P,\righths)}}{\XcolY{\Pneg_\lefths}{(P,\lefths)}}}{\strut}
            \qquad
            \infer{\DerDec{}{(l^-,\lefths)}{\XcolY{\Ppos_\lefths}{(l^-,\lefths)}}}{\strut}
            \\\\
            \infer{\DerDec{}{(\falseN,\lefths)}{\XcolY{\Ptrue_\lefths}{(\falseN,\lefths)}}}{\strut}
            \qquad
            \infer
                {\DerDec{}{\Delta\Dand{(\falseN,\lefths)}}{\XcolY{\switchr p}{(\neg A,\lefths)}}}
                {\DerDec{}{\Delta}{\XcolY{p}{(A,\righths)}}}
                \\\\
                  \infer{\DerDec{}{\Delta_1,\Delta_2}{\XcolY{\cons{p_1}{p_2}}{(A_1\imp A_2,\lefths)}}}{
                    \DerDec{}{\Delta_1}{\XcolY{p_1}{(A_1,\righths)}}
                    \quad
                    \DerDec{}{\Delta_2}{\XcolY{p_2}{(A_2,\lefths)}}
                  }
                  \qquad
                  \infer{\DerDec{}{\Delta}{\XcolY{\project i p}{(A_1\andN A_2,\lefths)}}}{\DerDec{}{\Delta}{\XcolY{p}{(A_i,\lefths)}}}
    \end{array}
    \]
  \caption{Decomposition relation for \LAF[J]}
  \label{fig:decompJ}
\end{bfigure}

Again, we can already see in Fig.~\ref{fig:decompJ} the traditional
right-introduction rules of positive connectives and left-introduction
rules of negative connectives, and again, it is clear from these rules
that \LAF[J] is well-founded.

We now interpret \LAF[J] sequents as intuitionistic sequents (from \eg
\LJF~\cite{liang09tcs}):
\begin{definition}[{\LAF[J]} sequents as two-sided \LJF\ sequents]

 First, when $\pm$ is either $+$ or $-$, we define
    \[
    \begin{array}l
      \Acc[\pm\righths]\Gamma\eqdef\{A\mid(A,\righths)\in\Acc[\pm]\Gamma\}\\
      \Acc[+\lefths]\Gamma\eqdef\{l^-\mid(l^-,\lefths)\in\Acc[+]\Gamma\}\\
      \Acc[-\lefths]\Gamma\eqdef\{N\mid(N,\lefths)\in\Acc[-]\Gamma\}
    \end{array}
    \]
 
    Then we define the encoding:
  \[\begin{array}{l@{\qquad\eqdef\qquad}l}
  \phi(\Der\Gamma {c})&\DerLJF {\Acc[+\righths]\Gamma,\Acc[-\lefths]\Gamma}{}{\Acc[+\lefths]\Gamma,\Acc[-\righths]\Gamma}\\
  \phi(\Der\Gamma {\XcolY {x^+} {(l^-,\lefths)}})
  &\DerFlLJF{\Acc[+\righths]\Gamma,\Acc[-\lefths]\Gamma} {l^-}{\Acc[+\lefths]\Gamma,\Acc[-\righths]\Gamma}\\
  \phi(\Der\Gamma {\XcolY {f} {\Drefute {(P,\righths)}}})
  &\DerFlLJF {\Acc[+\righths]\Gamma,\Acc[-\lefths]\Gamma}{P}{\Acc[+\lefths]\Gamma,\Acc[-\righths]\Gamma}\\
  \phi(\DerF\Gamma {\XcolY {t^+} {(N,\lefths)}}{})
  &\DerFlLJF{\Acc[+\righths]\Gamma,\Acc[-\lefths]\Gamma} {N}{\Acc[+\lefths]\Gamma,\Acc[-\righths]\Gamma}\\
  \phi(\Der\Gamma {\XcolY {x^+} {(l^+,\righths)}})
  &\DerFrLJF{\Acc[+\righths]\Gamma,\Acc[-\lefths]\Gamma} {l^+}\\
  \phi(\Der\Gamma {\XcolY {f} {\Drefute {(N,\lefths)}}})
  &\DerFrLJF {\Acc[+\righths]\Gamma,\Acc[-\lefths]\Gamma}{N}\\
  \phi(\DerF\Gamma {\XcolY {t^+} {(P,\righths)}}{})
  &\DerFrLJF{\Acc[+\righths]\Gamma,\Acc[-\lefths]\Gamma} {P}
  \end{array}
  \]

In the first four cases, we require $\Acc[+\lefths]\Gamma,\Acc[-\righths]\Gamma$ to be a singleton (or be empty).
\end{definition}

\begin{definition}[Typing contexts]\label{def:intuitionisticcontext}
 We assume that we always have $(\falseN,\lefths)\in\Acc[+]{\Gamma}$ and that
 \[
 \begin{array}c
   \begin{array}{ll@{\qquad}ll}
    \Acc[+]{\Gamma\Cextend (l^+,\righths)}&=\Acc[+]{\Gamma}\cup \{(l^+,\righths)\}
    &
    \Acc[-]{\Gamma\Cextend a}&=\Acc[-]{\Gamma}\\
    \Acc[+]{\Gamma\Cextend \Drefute M}&=\Acc[+]{\Gamma}
    &
    \Acc[-]{\Gamma\Cextend \Drefute{(N,\lefths)}}&=\Acc[-]{\Gamma}\cup \{(N,\lefths)\}\\
    \Acc[\pm]{\Gamma\Cextend\Dunit}&=\Acc[\pm]{\Gamma}
    &
   \Acc[\pm]{\Gamma\Cextend(\Delta_1\Dand\Delta_2)}&=\Acc[\pm]{\Gamma\Cextend\Delta_1\Cextend\Delta_2}
  \end{array}\\
   \begin{array}{ll}
    \Acc[+]{\Gamma\Cextend (v,\lefths)}
    &=\{(l^+,\righths)\mid (l^+,\righths)\in\Acc[+]{\Gamma}\}\cup\{(v,\lefths),(\falseN,\lefths)\}\\
    \Acc[-]{\Gamma\Cextend \Drefute{(P,\righths)}}
    &=\{(N,\lefths)\mid (N,\lefths)\in\Acc[-]{\Gamma}\}\cup\{(P,\righths)\}\\
  \end{array}
  \end{array}
  \]
  where again $\pm$ stands for either $+$ or $-$ and $v$ stands for either a negative literal $l^-$ or $\falseN$.
\end{definition}

The first three lines are the same as those assumed for $K1$, except
they are restricted to those cases where we do not try to add to $\Gamma$
an atom or a molecule that is interpreted as going to the right-hand
side of a sequent. When we want to do that, this atom or molecule
should overwrite the previous atom(s) or molecule(s) that was (were)
interpreted as being on the right-hand side; this is done in the last
two lines, where $\Acc[+\lefths]{\Gamma},\Acc[-\righths]{\Gamma}$ is
completely erased.

\begin{theorem}\label{th:adequacyJ}
  $\phi$ satisfies structural adequacy between \LAFcf[J]\ and \LJF.
\end{theorem}
The details are similar to those of Theorem~\ref{th:adequacyK1}, relying on the \LJF\ properties expressed in~\cite{liang09tcs}.

\section{Realisability semantics of \LAF}
\label{sec:RealTerms}
\label{sec:RealTypes}
\label{sec:RealAdequacy}

We now look at the semantics of \LAF\ systems, setting as an objective
the definition of models and the proof of their correctness at the
same generic level as that of \LAF.

In this section, $\powerset{\mathcal A}$ stands for the power set of
a given set $\mathcal A$.

Given a \LAF\ instance, we define the following notion
of realisability algebra:

\begin{definition}[Realisability algebra]
  A \Index{realisability algebra} is an algebra of the form
  \[
  \begin{array}l
  \left(
  \SPrim,\SPos,\SNeg, \orth \ \ , \SContexts,
       {\footnotesize
         \left(\begin{array}{c@{}l@{}l}
           \Data&\rightarrow&(\DecompType[\SPrim,\SNeg]\rightarrow\SPos)\\
           p&\mapsto&\spat p
         \end{array}\right)
       },
       {\footnotesize
         \left(\begin{array}{c@{}l@{}l}
           (\Data\pfunspace\PTerms)\times\SContexts&\pfunspace&\SNeg\\
           (f,\rho)&\mapsto&\sem[\rho] f
         \end{array}\right)
       },
       {\footnotesize
        \left(\begin{array}{c@{}l@{}l}
          \Atms&\rightarrow&\powerset{\SPrim}\\
          a&\mapsto&\sem{a}
        \end{array}\right)
       }       
  \right)
  \end{array}
  \]

  where
\begin{itemize}
\item $\SPrim,\SPos,\SNeg$ are three arbitrary sets of elements called
  \Index[label denotation]{label denotations},
  \Index[positive denotation]{positive denotations},
  \Index[negative denotation]{negative denotations},
  respectively;
\item $\orth \ \ $ is a relation between negative and positive
  denotations ($\orth \ \ \subseteq \SNeg\times\SPos$), called the
  \Index{orthogonality relation};
\item $\SContexts$ is a $(\SPrim,\SNeg)$-context algebra,
  whose elements, denoted $\rho,\rho',\ldots$, are called
  \Index[semantic context]{semantic contexts}.
\end{itemize}

The $(\SPrim,\SNeg)$-decomposition algebra $\DecompType[\SPrim,\SNeg]$
is abbreviated $\SDecs$; its elements, denoted $\frak\Delta$,
$\frak\Delta'$\ldots, are called \Index[semantic
decomposition]{semantic decompositions}.
\end{definition}

Given a model structure, we can define the interpretation of
proof-terms. The model structure already gives an interpretation for
the partial functions $f$ from patterns to commands. We extend it to
all proof-terms as follows:

\begin{definition}[Interpretation of proof-terms]

  Positive terms (in $\PTerms^+$) are interpreted as positive
  denotations (in $\SPos$), decomposition terms (in $\Decomp$) are
  interpreted as semantic decompositions (in
  $\SDecs$), and commands (in $\PTerms$)
  are interpreted as pairs in $\SNeg\times\SPos$ (that may or may not
  be orthogonal), as follows:
  \[
  \begin{array}{ll@{\qquad}ll@{}l@{}l}
    \sem[\rho]{pd}&\eqdef\spat p(\sem[\rho]d)
    &\sem[\rho]\Tunit&\eqdef\Dunit
    &\sem[\rho]{\cutc{x^-}{t^+}}&\eqdef(\varRead[x^-]\rho,\sem[\rho]{t^+})\\
    &&\sem[\rho]{d_1\Tand d_2}&\eqdef\sem[\rho]{d_1}\Dand\sem[\rho]{d_2}
    &\sem[\rho]{\cutc{f}{t^+}}&\eqdef(\sem[\rho]{\THO f},\sem[\rho]{t^+})\\
    &&\sem[\rho]{x^+}&\eqdef\varRead[x^+]\rho\\
    &&\sem[\rho]{\THO f}&\eqdef\sem[\rho]{\THO f}\mbox{ as given by the}&\mbox{ realisability}&\mbox{ algebra}\\
  \end{array}
  \]
\end{definition}

Our objective is now the Adequacy Lemma whereby, if $t$ is of type $A$
then the interpretation of $t$ is in the interpretation of $A$.  We
have already defined the interpretation of proof-terms in a model
structure.  We now proceed to define the interpretation of types.

In system \LAF, there are four concepts of ``type inhabitation'':
\begin{enumerate}
\item ``proving'' an atom by finding a suitable positive label in the
  typing context;
\item ``proving'' a molecule by choosing a way to decompose it into a
  typing decomposition;
\item ``refuting'' a molecule by case analysing all the possible ways
  of decomposing it into a typing decomposition;
\item ``proving'' a typing decomposition by inhabiting it with a decomposition term.
\end{enumerate}

Correspondingly, we have the four following interpretations, with the
interpretations of atoms (1.) in $\powerset{\SPrim}$ being arbitrary and
provided as a parameter of a realisability algebra:

\begin{definition}[Interpretation of types and typing contexts]
  Assume the instance of \LAF\ is well-founded. We define (by
  induction on the well-founded ordering between molecules):\\
  2. the positive interpretation of a molecule in $\powerset{\SPos}$;\\ 
  3. the negative interpretation of a molecule in $\powerset{\SNeg}$;\\ 
  4. the interpretation of a typing decomposition in $\powerset{\DecompType[\SPrim,\SNeg]}$:

  \[\begin{array}{llll}
  \SemTyP {M}  &\eqdef \{ \spat p (\frak \Delta)\in\SPos
  &\mid \frak\Delta\in\SemTy{\Delta},\mbox{ and } \DerDec{}{\Delta} {\XcolY p {M}}\}\\[5pt]
  \SemTyN {M}  &\eqdef \{ \frak n\in \SNeg
  &\mid \forall \frak p \in\SemTyP M , \orth{\frak n}{\frak p}\}\\[5pt]
  \SemTy {\Dunit} &\eqdef\{\Dunit\}\\
  \SemTy {\Delta_1,\Delta_2} &\eqdef\{\frak \Delta_1,\frak \Delta_2 
  &\mid \frak \Delta_1\in\SemTy{\Delta_1}\mbox{ and }\frak \Delta_2\in\SemTy{\Delta_2}\}\\
  \SemTy {a} &\eqdef\SemTy {a}&\mbox{ as given by the realisability algebra}\\
  \SemTy {\Drefute M} &\eqdef\{\Drefute\frak n &\mid \frak n\in\SemTyN {M}\}
  \end{array}
  \]

  We then define the interpretation of a typing
  context:% \footnote{In this definition we implicitly require that
    % $\domP\rho=\domP\Gamma$, $\domN\rho=\domN\Gamma$ and for all
    % $x^+\in\domP\rho$ (\resp $x^-\in\domN\rho$)
    % $\SemTy{\varRead[x^+]\Gamma}$ (\resp
    % $\SemTyN{\varRead[x^-]\Gamma}$) is defined.}
  \[\begin{array}{llll}
  \SemTy\Gamma\eqdef\{\rho\in\SContexts\mid
  &\forall x^+\in\domP\rho,\ \varRead[x^+]\rho\in\SemTy{\varRead[x^+]\Gamma}\\
  &\forall x^-\in\domN\rho,\ \varRead[x^-]\rho\in\SemTyN{\varRead[x^-]\Gamma}
  &\}
  \end{array}\]
\end{definition}

Now that we have defined the interpretation of terms and the interpretation of types, we get to the Adequacy Lemma.

\begin{lemma}[Adequacy for \LAF]\label{lem:adequacy}
  We assume the following hypotheses:
  \begin{enumerate}
  \item[Well-foundedness:]
    The \LAF\ instance is well-founded.
  \item[Typing correlation:]
    If $\rho\in\SemTy\Gamma$ and $\frak \Delta\in\SemTy{\Delta}$ then $(\rho\Cextend\frak\Delta)\in\SemTy{\Gamma\Textend{\Delta}}$.
  \item[Stability:]
    If $\frak d\in\SemTy{\Delta}$ for some $\Delta$ and $\SemTe{f(p)}{\rho\Cextend \frak d}\in\orth{}{}$, then $\orth{\sem[\rho] f}{\spat p(\frak d)}$.
  \end{enumerate}

  We conclude that, for all $\rho\in\SemTy\Gamma$,
  \begin{enumerate}
  \item if $\DerF\Gamma {\XcolY {t^+} {M}}{}$ then $\SemTe{t^+}\rho\in\SemTyP {M}$; 
  \item if $\Der\Gamma {\XcolY {d} {\Delta}}{}$ then $\SemTe{d}\rho\in\SemTy {\Delta}$; 
  \item if $\Der\Gamma {{t} {}}{}$ then $\SemTe{t}\rho\in\orth{}{}$.
  \end{enumerate}
\end{lemma}
\begin{proof}See the proof in \Coq~\cite{LengrandHDRCoq}.\end{proof}

Looking at the Adequacy Lemma, the stability condition is traditional:
it is the generalisation, to that level of abstraction, of the usual
condition on the orthogonality relation in orthogonality models (those
realisability models that are defined in terms of orthogonality,
usually to model classical
proofs~\cite{girard-ll,DanosKrivine00,Krivine01,MunchCSL09}):
orthogonality is ``closed under anti-reduction''. Here, we have not
defined a notion of reduction on \LAF\ proof-terms, but
intuitively, we would expect to rewrite
$\cutc{\THO f}{pd}$ to $f(p)$ ``substituted by $d$''.

On the other hand, the typing correlation property is new, and is due
to the level of abstraction we operate at: there is no reason why our
data structure for typing contexts would relate to our data structure
for semantic contexts, and the extension operation, in both of them,
has so far been completely unspecified. Hence, we clearly need such an
assumption to relate the two.

However, one may wonder when and why the typing correlation property
should be satisfied. One may anticipate how typing correlation could
hold for the instance \LAF[K1] of
\LAF. Definition~\ref{def:classicalcontext} suggests that, in the
definition of a typing context algebra, the extension operation does
not depend on the nature of the atom $a$ or molecule $M$ that is being
added to the context. So we could \emph{parametrically} define
$(\mathcal A,\mathcal B)$-contexts for any sets $\mathcal A$ and
$\mathcal B$ (in the sense of relational
parametricity~\cite{Rey:typapp}). The typing context algebra would be
the instance where $\mathcal A=\Atms$ and $\mathcal B=\Moles$ and the
semantic context algebra would be the instance where $\mathcal
A=\SPrim$ and $\mathcal B=\SNeg$. Parametricity of context extension
would then provide the typing correlation property.

\section{Example: boolean models to prove Consistency}
\label{sec:RealConsistency}

We now exhibit models to prove the consistency of \LAF\ systems.

%% Assume we have a \LAF\ instance
%% \[(\Sorts,\Terms,\SoContexts,\decs{},\Atms,\Moles,\atmEq{}{},\var[+],\var[-],\TContexts,\Data,\decf{})\]

%% \begin{theorem}[Consistency of \LAF\ instances]

%%   Assume we have a realisability algebra where
%%   \begin{itemize}
%%   \item $\orth \ \ = \emptyset$;
%%   \item there is a semantic context $\rho_\emptyset$ with $\domP{\rho_\emptyset}=\domN{\rho_\emptyset}=\emptyset$;
%%   \item if $\rho\in\SemTy\Gamma$ and $\frak \Delta\in\SemTy{(\Delta^l,{\bf r})}$ then $(\rho\Cextend\frak\Delta)\in\SemTy{\Gamma\Textend{\Delta^l}}$.
%%   \end{itemize}
%%   Then there is no typing context $\Gamma_\emptyset$ with $\domP{\Gamma_\emptyset}=\domN{\Gamma_\emptyset}=\emptyset$ and command $t$ such that
%%   $\Der{\Gamma_\emptyset} {{t} {}}{}$.
%% \end{theorem}
%% \begin{proof}
%%   Stability is obvious. If there was such a $\Gamma_\emptyset$ and $t$, then we would have $\rho_\emptyset\in\SemTy{\Gamma_\emptyset}$, and the Adequacy Lemma (Lemma~\ref{lem:adequacy}) would conclude $\SemTe t{{\rho_\emptyset}}\in\emptyset$.
%% \end{proof}

We call \Index{boolean realisability algebra} a realisability algebra
where $\orth\ \ =\emptyset$. The terminology comes from the fact
that in such a realisability algebra, $\SemTyN{M}$ can only
take one of two values: $\emptyset$ or $\SNeg$, depending on whether
$\SemTyP{M}$ is empty. A boolean realisability algebra
satisfies Stability.

\begin{theorem}[Consistency of \LAF\ instances]
  Assume we have a well-founded \LAF\ instance, and a boolean
  realisability algebra for it, where typing correlation holds and there is an empty semantic context $\rho_\emptyset$.
  There is no empty typing context $\Gamma_\emptyset$ and command $t$ such that
  $\Der{\Gamma_\emptyset} {{t} {}}{}$.
\end{theorem}
\begin{proof}
  The previous remark provides Stability. If there was such a
  $\Gamma_\emptyset$ and $t$, then we would have
  $\rho_\emptyset\in\SemTy{\Gamma_\emptyset}$, and the Adequacy Lemma
  (Lemma~\ref{lem:adequacy}) would conclude $\SemTe
  t{{\rho_\emptyset}}\in\emptyset$.
\end{proof}

We provide such a realisability model that works with all ``parametric'' \LAF\ instances:

\begin{definition}[Trivial model for parametric \LAF\ instances]

  Assume we have a parametric \LAF\ instance, \ie an instance where
  the typing context algebra $\TContexts$ is the instance
  $\Contexts[\Atms,\Moles]$ of a family of context algebras
  $(\Contexts[\mathcal A,\mathcal B])_{\mathcal A,\mathcal B}$ whose
  notion of extension is defined parametrically in $\mathcal
  A,\mathcal B$.
  The \Index{trivial boolean model} for it is:
  \columns{
    \column[.49]{
      \[
      \begin{array}{rc@{\,}l}
        \SPrim\eqdef\quad\SPos\eqdef&\SNeg&\eqdef\unitt\\
        &\orth{}{}&\eqdef \emptyset\\
        \forall \frak{\Delta}\in\SDecs,&\spat p(\frak{\Delta})&\eqdef \uniti\\
        \forall f:\Data\pfunspace\PTerms, \forall \rho\in\SContexts,& \sem[\rho] f&\eqdef \uniti\\
        \forall a\in\Atms,
        &\sem {a}
        &\eqdef\unitt
      \end{array}
      \]
    }%\vrule
    \cquad
    \cquad
    \cquad
    \cquad
    \cquad
    \hfill
    \column{
      \[
      \begin{array}{llll}
        \SContexts\eqdef \Contexts[\unitt,\unitt]\\
        \mbox{and therefore}\\
        \begin{array}{llll}
          \forall\rho\in\SContexts,\forall x^+\in\domP\rho,& \varRead[x^+]\rho&\eqdef \uniti\\
          \forall\rho\in\SContexts,\forall x^-\in\domN\rho,& \varRead[x^-]\rho&\eqdef \uniti\\
        \end{array}
      \end{array}
      \]
    }
  }
\end{definition}

Note that, not only can $\SemTyN{M}$ only take one of the two
values $\emptyset$ or $\unitt$, but $\SemTyP{M}$ can also only
take one of the two values $\emptyset$ or $\unitt$.

We can now use such a structure to derive consistency of parametric \LAF\ instances:

\begin{corollary}[Consistency for parametric \LAF\ instances]\label{cor:consistency}
Assume we have a parametric \LAF\ instance that is well-founded and assume there is an empty $(\unitt,\unitt)$-context in $\Contexts[\unitt,\unitt]$.
Then there is no empty typing context $\Gamma_\emptyset$ and command $t$ such that
$\Der{\Gamma_\emptyset} {{t} {}}{}$.

In particular, this is the case for \LAF[K1].
\end{corollary}

The system \LAF[J] does not fall in the above category since the
operation of context extension is not parametric enough: when
computing $\Gamma\Textend{a}$ (\resp $\Gamma\Textend{\Drefute{M}}$),
we have to make a case analysis on whether $a$ is of the form
$(l^+,\righths)$ or $(v,\lefths)$ (\resp whether $M$ is of the form
$(N,\lefths)$ or $(P,\righths)$).

But we can easily adapt the above trivial model into a
not-as-trivial-but-almost model for \LAF[J], as is shown
in~\cite{LengrandHDR}, Ch.~6.

\section{Conclusion and Further Work}
\label{Conclusion}

\subsection{Contributions}

In this paper we have used, and slightly tweaked, a system with
proof-terms proposed by Zeilberger for \emph{big-step
  focussing}~\cite{ZeilbergerPOPL08}, which abstracts away from the
idiosyncrasies of logical connectives and (to some extent) logical
systems: In particular we have shown how two focussed sequent calculi
of the literature, namely \LKF\ and \LJF~\cite{liang09tcs}, are
captured as instances of this abstract focussing system \LAF.

Building on Munch-Maccagnoni's description~\cite{MunchCSL09} of
classical realisability in the context of polarised classical logic
and focussing, we have then presented the realisability models for
\LAF, thus providing a generic approach to the denotational semantics
of focussed proofs. Central to this is the Adequacy
Lemma~\ref{lem:adequacy}, which connects typing and realisability, and
our approach is generic in that the Adequacy Lemma is proved once and
for all, holding for all focussed systems that can be captured as
\LAF\ instances.

Incidently, a by-product of this paper is that we provided proof-term
assigments for \LKF\ and \LJF, and provided realisability semantics
for their proofs. We believe this to be new.

But showing the Adequacy Lemma at this level of abstraction was also
motivated by the will to exhibit how close typing and realisability
are while differing in an essential way:

\subsection{Typing vs. realisability}

Concerning the \emph{positives}, typing and realisability simply mimic
each other: Ignoring contexts,
\begin{itemize}
  \item in typing, $\DerF{} {\XcolY {t^+} M}{}$ means $t^+$ is
of the form $pd$ with $\DerDec{}{\Delta}{\XcolY p{M}}$ and
$\Der{}{\XcolY d{\Delta}}$ for some $\Delta$;
\item
in realisability,
$\sem{t^+}\in\SemTyP {M}$ means $\sem{t^+}$ is of the form $\spat
p(\frak d)$ with $\DerDec{}{\Delta}{\XcolY p{M}}$ and ${\frak
  d}\in\Delta$ for some $\Delta$.
\end{itemize}

Concerning the \emph{negatives}, it is appealing to relate the
quantification in rule \textsf{async} with the quantification in the
orthogonality construction:
\begin{itemize}
  \item in typing, ${\Der{} {\XcolY {\THO f} {\Drefute M}}}$ means that for
all $p$ and $\Delta$ such that $\DerDec{}{\Delta}{\XcolY p{M}}$, we
have $\Der{\Cextend[x]\Delta}{f(p)}$\linebreak ($\Delta$ extending the empty
typing context);
\item
in realisability, $\sem f\in\SemTyN {M}$ means that for all
$p$ and $\Delta$ such that $\DerDec{}{\Delta}{\XcolY p{M}}$, for all
$\frak \Delta\in\SemTy{\Delta}$ we have $\orth{\sem f}{\spat p (\frak
  \Delta)}$, usually obtained from $\sem[\Cextend\frak \Delta]
{f(p)}\in\orth{}{}{}$ ($\frak\Delta$ extending the empty
semantic context).
\end{itemize}

In both cases, a ``contradiction'' needs to be reached for all $p$ and
$\Delta$ decomposing $M$. But in typing, the proof-term $f(p)$
witnessing the contradiction can only depend on the pattern $p$ and
must treat its holes (whose types are agregated in $\Delta$)
parametrically, while in realisability, the reason why $\orth{\sem
  f}{\spat p (\frak\Delta)}$ holds, though it may rely on the
computation of $f(p)$, can differ for every $\frak
\Delta\in\SemTy{\Delta}$.  It is the same difference that lies between
the usual rule for $\forall$ in a Natural Deduction system for
arithmetic, and the $\omega$-rule~\cite{hilbert31,schuette50}:
\[
\infers{\forall n A}[\forall\mbox{-intro}]{A}\qquad \infers{\forall n A}[\omega]{\subst A n 0\quad \subst A n 1\quad\subst A n 2\quad\ldots}
\]
The difference explains why typing is (usually) decidable, while
realisability is not, and contributes to the idea that ``typing is but
a decidable approximation of realisability''.

We believe that we have brought typing and realisability as close as
they could get, emphasising where they coincide and where they differ,
in a framework stripped from the idiosyncrasies of logics, of their
connectives, and of the implementation of those functions we use for
witnessing refutations, \ie inhabiting negatives (\eg the $\lambda$ of
$\lambda$-calculus).

\subsection{Coq formalisation and additional results}

In this paper we have also exhibited simple realisability models to
prove the consistency of \LAF\ instances.

The parameterised \LAF\ system has been formalised in
Coq~\cite{LengrandHDRCoq}, together with realisability algebras. The
Adequacy Lemma~\ref{lem:adequacy} and the Consistency result
(Corollary~\ref{cor:consistency}) are proved there as well.
Because of the abstract level of the formalism, very few concrete
structures are used, only one for $(\mathcal A,\mathcal
B)$-decompositions and one for proof-terms; rather, Coq's
\emph{records} are used to formalise the algebras used throughout the
paper, declaring type fields for the algebras' support sets, term
fields for operations, and proof fields for specifications.  Coercions
between records (and a few structures declared to be canonical) are
used to lighten the proof scripts.

Besides this, the Coq formalisation presents no major surprises. It
contributes to the corpus and promotion of machine-checked theories
and proofs. However, formalising this in Coq was a particularly
enlightening process, directly driving the design and definitions of
the concepts. In getting to the essence of focussing and stripping the
systems from the idiosyncrasies of logics and of their connectives,
Coq was particularly useful: Developing the proofs led to identifying
the concepts (\eg what a typing context is), with their basic
operations and their minimal specifications.  Definitions and
hypotheses (\eg the three hypotheses of the Adequacy Lemma) were
systematically weakened to the minimal form that would let proofs go
through. Lemma statements were identified so as to exactly fill-in the
gaps of inductive proofs, and no more.

The formalisation was actually done for a \emph{first-order} version
of \LAF, that is fully described in~\cite{LengrandHDR}. That in itself
forms a proper extension of Zeilberger's
systems~\cite{ZeilbergerPOPL08}. In this paper though we chose to
stick to the propositional fragment to simplify the presentation.

Regarding realisability models, more interesting examples than those
given here to prove consistency can obviously be built to illustrate
this kind of semantics.
In particular, realisability models built from the term syntax can be
used to prove normalisation properties of \LAF, as shown
in~\cite{LengrandHDR}.  Indeed, one of the appeals of big-step
focussing systems is an elegant notion of cut-reduction, based on
rewrite rules on proof-terms and with a functional programming
interpretation in terms of \emph{pattern-matching}.
A cut-reduction system at the abstraction level of \LAF\ is given
in~\cite{LengrandHDR}, in terms of an abstract machine (performing
head-reduction). A command $t$ is evaluated in an evaluation context
$\rho$; denoting such a pair as $\machine[\rho]t$, we have the main
reduction rule (where $d'$ stands for the evaluation of $d$ in the evaluation context $\rho$):
\[\machine[\rho]{\cutc{\THO f}{pd}}\Rew{}\machine[\rho\Cextend d']{f(p)}\]

Normalisation of this system (for a well-founded \LAF\ instance), is
proved by building a syntactic realisability model, in which
orthogonality holds when the interaction between a negative denotation
and a positive one is normalising. This model, together with the head
normalisation result, are also formalised in
Coq~\cite{LengrandHDRCoq}.
It forms a formal connection, via the
orthogonality techniques, between proofs of normalisation \emph{\`a
  la} Tait-Girard and realisability.
From this termination result, an informal argument is
proposed~\cite{LengrandHDR} to infer cut-elimination, but the argument
still needs to be formalised in Coq. This is tricky since,
cut-elimination needing to be performed arbitrarily deeply in a
proof-tree (``under lambdas''), we need to formalise a notion of
reduction on those functions we use for witnessing refutations, for
which we have no syntax.

Finally, more work needs to be done on formalising the connections
between \LAF\ and other related systems: firstly, \LAF\ is very
strongly related to \emph{ludics}~\cite{Girard01}, a system for
big-step focussing for linear logic, and which is also related to game
semantics. \LAF\ can be seen as a non-linear variant of ludics, our
proof-term-syntax more-or-less corresponding to ludics'
\emph{designs}. But in order to get linearity, \LAF\ would need to
force it in the typing rules for the decomposition terms $x^+$,
$\Tunit$, and $d_1\Tand d_2$. It would also be interesting to
investigate whether or how \LAF\ could be adapted to modal logics.

\bibliographystyle{good}
\bibliography{Common/abbrev-short,Common/Main,Common/crossrefs}

\end{document}